\documentclass[11pt,a4paper]{article}

\usepackage{upgreek}
\usepackage{amsmath}
\usepackage{amssymb,a4wide}

\newtheorem{theorem}{Theorem}[section]

\newenvironment{proof}[1][Proof]{\textbf{#1.} }{\ \rule{0.5em}{0.5em}}

\newcommand{\eps}{\upvarepsilon}

\begin{document}

\title{Lower bounds on the performance of online algorithms for relaxed packing problems}

\author{J\'anos Balogh\thanks{Institute of Informatics,
		University of Szeged, Szeged, Hungary. \texttt{baloghj@inf.u-szeged.hu}.}      \and
Gy\"{o}rgy D\'{o}sa\thanks{Department of
Mathematics, University of Pannonia, Veszpr\'em, Hungary.
\texttt{dosagy@almos.vein.hu}.}
\and
Leah Epstein\thanks{
Department of Mathematics, University of Haifa, Haifa, Israel.
\texttt{lea@math.haifa.ac.il}.} \and   {\L}ukasz
Je{\.z}\thanks{Institute of Computer Science, University of
Wroc{\l}aw, Wroc{\l}aw,
Poland. \texttt{lje@cs.uni.wroc.pl}. {{\L}. Je{\.z} was supported by the Polish National Science Center grant 2020/39/B/ST6/01679.}}}

\title{Lower bounds on the performance of online algorithms \\ for relaxed packing problems}

\date{}

%\maketitle

\maketitle              % typeset the header of the contribution

\begin{abstract}
We prove new lower bounds for suitable competitive ratio measures of two relaxed online packing problems: online removable multiple knapsack, and a recently introduced online minimum peak appointment scheduling problem.  The high level objective in both problems is to pack arriving items of sizes at most $1$ into bins of capacity $1$ as efficiently as possible,  but the exact formalizations differ.  In the appointment scheduling problem, every item has to be assigned to a position, which can be seen as a time interval during a workday of length $1$. That is, items are not assigned to bins, but only once all the items are processed, the optimal number of bins subject to chosen positions is determined, and this is the cost of the online algorithm.  On the other hand, in the removable knapsack problem there is a fixed number of bins, and the goal of packing items, which consists in choosing a particular bin for every packed item (and nothing else), is to pack as valuable a subset as possible. In this last problem it is possible to reject items, that is, deliberately not pack them, as well as to remove packed items at any later point in time, which adds flexibility to the problem.
\end{abstract}
\noindent{\bf keywords: Bin packing, online algorithms, competitive ratio}

%\begin{document}

\section{Introduction}

We study two online problems, for which the offline version is a classic problem, with well-known efficient near-optimal solutions \cite{IK75,FerLue81,KK82}. Our online problems are not the most natural variants that one can define, but they are more relaxed. This models reality in the sense that often there is some flexibility even in online environments. Flexibility obviously allows the design of better online algorithms, though such algorithms typically cannot find optimal solutions. Here, we focus on the limitations of online algorithms for relaxed models.

We use the competitive ratio and asymptotic competitive ratio measures for analysis of online algorithms. The competitive ratio for minimization problems is the worst-case ratio between the cost of an online algorithm and the cost of an optimal offline algorithm for the same input. For maximization problems, the roles of the algorithm and an optimal offline solution are reversed. The asymptotic competitive ratio is the supreme limit of the competitive ratio for inputs with optimal costs or profits growing to infinity.

In offline bin packing \cite{JoDUGG74,FerLue81,KK82}, there are items of indices $1,2,\ldots,n$, where item $j$ has a rational size $s_j\in(0,1]$. The goal is to partition the items into the minimum number of sets called bins, where the total size for every bin does not exceed $1$. One can see this as a scheduling problem where items are assigned to machines that are available during the time interval $[0,1)$, but it is not necessary to assign the specific time slots in advance, since this can always be done. The length of the time slot for item $j$ is required to be of length $s_j$, where the interval has the form $[x,x+s_j)$ for $x\geq 0$ and $x\leq 1-s_j$. Alternatively, one can assign the time slots, and not the bins, where the assignment to bins can be done by a simple process of coloring an interval graph. In the standard online bin packing problem \cite{BBG,BBDEL_ESA18,BBDEL_newlb} items are to be assigned to bins sequentially, and it is assumed that items just receive consecutive time slots in the bin, starting from time zero. An alternative online model was defined recently by Escribe, Hu, and Levi \cite{EHL21}, where the online feature is the assignment to time slots rather than the bins. The problem is called minimum peak appointment scheduling (MPAS). In both models, items are presented one by one, such that each item is assigned irrevocably before the next item arrives.

In the work by Escribe, Hu, and Levi \cite{EHL21}, a randomized algorithm with asymptotic competitive ratio at most $1.5$ was designed, which was recently improved to~$\frac{16}{11}\approx 1.455$ by Smedira and Shmoys~\cite{SS21}. A lower bound of $1.5$ on the competitive ratio of deterministic algorithms was proved \cite{EHL21}, while Smedira and Shmoys~\cite{SS21} proved a lower bound of $1.2$ for the asymptotic competitive ratio of all randomized (and deterministic) algorithms.  These results contrast with those known for the standard online bin packing. While the best known lower bound of 1.54278 \cite{BBDEL_newlb} holds only for deterministic algorithms, earlier results hold for randomized algorithms, where the best such result is 1.5403 \cite{BBG}. The current best upper bound for standard online bin packing~\cite{BBDEL_ESA18} is 1.57829. A simple way to see the difference between the problems is the following. Consider a large even number of items of size $0.4$, possibly followed by the same number of items of size $0.6$. A bin packing algorithm has to decide how many pairs of items of size $0.4$ to create in order to have a good performance in both cases. However, in the case of MPAS, one can assign half of the items of size $0.4$ to the time interval $[0,0.4)$ and the other half to $[0.6,1)$. This is optimal if there are no further items, but if there are such items, half of them is assigned to $[0,0.6)$ and the other half to $[0.4,1)$, also obtaining an optimal solution. The problem is not meaningful as a separate offline problem, though generalizations were recently studied as offline problems \cite{RKS17,GGJK21,DJKRT21}.

A related problem is the so called \emph{dual bin packing}~\cite{BFLN01,ABFLNE02}, which may also be seen as a variant of the \emph{knapsack problem} with multiple knapsacks available~\cite{CJS16,IT02}: The processing of arriving items is similar, except here the number of available bins (or knapsacks) is fixed in advance, and the goal is rather to maximize the profit associated with those items that are successfully packed.
In the multiple knapsack as studied originally, the profit associated with a set of packed items was the \emph{maximum} over all the bins of the total value of all items packed in the bin~\cite{IT02}, but later studies~\cite{CJS16,BPP20} extended this to the \emph{sum} of values of all the items packed, which is in line with dual bin packing.  We are interested in this objective, so we will not specify the results concerned solely with the single bin of maximum value.  The packing consists in assigning an item to a particular bin, which is in contrast to MPAS, where instead a ``position'' or ``interval'' (on a horizontal axis) within a bin (which is yet to be determined) is specified.  An item may also be rejected by an algorithm, i.e., not packed at all.  Moreover, in the so called removable online variant of the knapsack problem it is allowed to remove a previously packed item (e.g., to accommodate the one arriving ), which from then on counts as rejected.  As is the case in vast amount of literature, we will consider two restricted settings, in which there is a particular natural relation between the profit associated with an item and its size (or processing time); note that as we focus on lower bounds, considering these makes our results stronger.  The two cases, which we call as in~\cite{CJS16}, are
\emph{proportional}, in which the value of an item equals its size, and \emph{unit}, in which every item is worth $1$ regardless of its size.

The dual bin packing problem corresponds to the unit case with no removals, for which no algorithm can attain constant competitive ratio~\cite{BFLN01}.
Thus the studies of this problem focused on ``accommodating'' instances, in which all items can be packed by the offline solution, for which
constant-competitive algorithms were designed.  Moreover, it is known that whether an algorithm is allowed to reject an item that it could pack in some bin (thus being ``unfair'') affects what ratio can be attained~\cite{ABFLNE02}.  In the later studies of the multiple knapsack problem, it was noted that proportional instances, even non-accommodating, allow constant-competitive ratio, which was eventually determined to be exactly $1+\ln 2\approx 1.6903$~\cite{CJS16,BPP20}.  Moreover, with removals allowed, a deterministic algorithm of asymptotic competitive ratio at most $3$ is known even for general instances, and the proportional and unit instances admit algorithms with much better competitive ratios of $1.6$ and $1.5$ respectively~\cite{CJS16}.  The corresponding lower bounds for these two settings, applicable even to randomized algorithms are only $\frac{8}{7}\approx 1.14$~\cite{CJS16} and $\frac{7}{6}\approx 1.17$~\cite{ABFLNE02} respectively.  No better lower bounds are known for general instances. For special cases with a small number of bins and the proportional case, lower bounds of $\frac 43$ and $\frac 65$ are known for the competitive ratio of deterministic algorithms with two and three bins, respectively \cite{CJS16}, and in some special cases the lower bound on the competitive ratio as a function of the number of bins is slightly inferior to the bounds of $\frac87$ and $\frac 76$ \cite{CJS16,ABFLNE02}.
This problem is also not of interest as a separate offline problem, though the knapsack problem and its variants are being studied continuously, and a near-optimal solution is known for almost fifty years \cite{IK75}.

\subsection{Our results}
In this work, we prove the following lower bounds on the performance of online algorithms:
\begin{itemize}
	\item A lower bound of $1.2287$  on the competitive ratio for deterministic algorithms for either the proportional or the unit case of the removable knapsack problem, improving upon the previous bounds of $\frac{8}{7}\approx 1.14$~\cite{CJS16} and $\frac{7}{6}\approx 1.17$~\cite{ABFLNE02} respectively,
	\item A lower bound of $1.2$ on the competitive ratio for randomized algorithms for the proportional case of the removable knapsack problem, also improving upon the previous bound of $\frac{8}{7}\approx 1.14$~\cite{CJS16},
	\item A lower bound of $1.2691$ on the asymptotic competitive ratio for randomized algorithms for the minimum peak appointment scheduling problem, improving upon the previous bound of $1.2$~\cite{SS21}.
\end{itemize}
We also consider some special cases for the online removable knapsack problem. For example, our results improve the known lower bound for three bins and deterministic algorithms.

\subsection{Adaptive item sizing in designing hard instances for online algorithms}
Those of our lower bounds that are designed specifically for deterministic algorithms employ the ``adaptive item sizing'' technique~\cite{BBDEL_newlb}, which we now describe.
This approach and more advanced approaches (in particular, one where there is a multiplicative gap between sizes) were used for other online bin packing problems \cite{FK13,BBDEL_ESA,balogh2017lower}.

This is a procedure in which a sequence $S$ of items arrives, all with sizes in a predetermined interval $[\alpha,\beta]$, where $\alpha$ and $\beta$ are parameters. This allows the items from $S$ to be partitioned into a set of \emph{smallish} items with sizes in the interval $[\alpha,\theta)$ and \emph{largish} items with sizes in the interval $(\theta,\beta]$. The threshold $\theta$ satisfies $\alpha<\theta<\beta$, and the classification of each item as either smallish or largish occurs immediately after its packing by the deterministic algorithm, but the value of $\theta$ is determined later.

Such classification and partitioning can be ensured by a procedure resembling the binary search. Let $a$ and $b$ be variables such that $\alpha \leq a < b \leq \beta$. All items thus far classified as smallish have sizes in $[\alpha,\ a ]$, and all items thus far classified as largish have sizes in $[ b,\ \beta ]$. Initially, we let $a=\alpha$ and $b=\beta$.  Then, the next item to arrive can have any size in $(a,b)$, e.g., $\frac{a+b}{2}$.
Once it is dealt with (packed or rejected) by the algorithm, and hence classified as either smallish or largish, the value of $a$ or $b$ respectively is set to this item's size.  Once all items in $S$ are processed, $\theta$ can take any value that could be the size of a next item in the sequence and thus separates the sizes of smallish and largish items, e.g., $\theta=\frac{a+b}{2}$ for the final values of the variables $a$ and $b$.

\section{Online removable knapsack}

In our lower bound proofs, we assume without loss of generality that the online algorithm is lazy in that it defers removing (or rejecting) items as long as the packing is valid.  Namely, we assume that upon the arrival of an item $e$,
the algorithm decides to either reject it outright without packing it or chooses any bin $B$ to which $e$ is packed.  If the bin $B$ then overflows, the algorithm chooses a subset of items from $B$, excluding $e$, for removal; this subset has to be minimal such that its removal makes the total size of the items remaining in bin $B$ at most $1$.  Additionally, the algorithm is only allowed to reject the item $e$ outright if there is no bin where it fits without removals; note that this does not mean no removals take place when such bin exists, as the algorithm can choose to pack $e$ in another bin.

\subsection{Deterministic online algorithms for the proportional and unit cases}

We start with a deterministic lower bound.

\begin{theorem}\label{122}
Any deterministic online algorithm that works for any number of bins $k \geq 2$, has competitive ratio of at least $1.228713$.
\end{theorem}
\begin{proof}
We focus on the proportional case, and remark that the proof also applies to the unit case because all items in the strategy have sizes that deviate from $\frac{1}{2}$ only by negligible amounts.

The input consists of two phases.  The first phase employs the adaptive sizing framework for $k$ items, with $\frac{1}{2}-\eps < \alpha < \beta < \frac{1}{2}$ for arbitrarily small $\eps>0$.  As there are $k$ bins, the algorithm is lazy, and item sizes are slightly below $\frac{1}{2}$, every item will be packed upon arrival (though possibly causing removal of another item), every non-empty bin at all times will contain either one or two items, and the number of items in any bin cannot decrease over item.  The items are classified as follows upon packing: The first item to be placed in a bin is largish, the second one to be placed in a bin is smallish, and an item that replaces another inherits the replaced item's class.

As a result, at the end of the first phase, each bin has at most two items, if is it non-empty, then it contains one largish item, and if it contains another one, then that item is smallish.  Let $\Gamma$ denote the number of bins with two items, where $\Gamma \leq \lfloor \frac k2 \rfloor$, since $k$ items were presented. If $\Gamma=0$, there are no smallish items, but the threshold $\theta$ is still well defined.

In the second phase, there are two possible continuations of the input, and the adversary's choice depends on $\Gamma$. The first one, applied when $\Gamma$ is suitably large, is to issue $k$ items of size $1-\beta > \frac{1}{2}$ each. Since all previous items have sizes at most $\beta$, the optimal offline solution packs these items in pairs, and its profit is at least $k\cdot (\frac 12-\eps)+k\cdot (1-\beta) > k \cdot (1-\eps)$.  Now consider the algorithm's packing.  As each item in the second phase has size strictly larger than $\frac{1}{2}$, any bin that was empty at the beginning of this phase, may contain at most one item at its end.
The number of bins for the algorithm with exactly one item before the arrival of the new items is at most $k-2\cdot \Gamma$, and therefore the profit of the algorithm is below $$(1-\beta)\cdot (k+\Gamma+(k-2\cdot \Gamma))<(\frac 12 + \eps)\cdot(2k-\Gamma) \ . $$ For $\eps \to 0$, the competitive ratio in this case tends to $\frac{2k}{2k-\Gamma}$. %In the next case we will assume that $\beta_1$ is very small and we just write the limit of the sizes for $\beta_1$ tending to zero.

The other strategy for the second phase, applied when $\Gamma$ is suitably small, is to issue $\Gamma+\lfloor \frac{k-\Gamma}2 \rfloor$ items of size $1-\theta$.  Each such item has size slightly above $\frac 12$, and can be packed together with a smallish item but not with a largish item.  Clearly, no bin can have more than one such item.  Note that the number of smallish items in the input is at least $\Gamma$, and thus the number of largish items is at most $k-\Gamma$, because exactly $k$ items were issued in the first phase.  It is possible that the number of smallish items is larger than $\Gamma$ if the algorithm removed a smallish item to pack another, which is then smallish as well.
Offline, it is possible to pack $\Gamma$ bins by placing one smallish item together with one of size $1-\theta$ in each bin,
$\lfloor \frac{k-\Gamma}2 \rfloor$ bins with the remaining items of size $1-\theta$, one per bin, and packing all the
now remaining items, i.e., largish and yet unpacked smallish ones in pairs into $\lceil \frac{k-\Gamma}2 \rceil$ bins; note that if $k-\Gamma$ is odd, one of those last bins will contain only a single item.
Again assuming that $\eps\to 0$, in the limit the total size of items packed this way is
$$ \Gamma + \frac{1}{2}\left\lfloor \frac{k-\Gamma}2 \right\rfloor + \left\lfloor \frac{k-\Gamma}2 \right\rfloor + \frac{1}{2}(k - \Gamma \mod 2)  \ .$$

For the online algorithm, an item of size $1-\theta$ cannot be added to a bin with only a single item, since such item is largish by the adaptive item sizing used by the adversary. So the algorithm has only $\Gamma$ bins with total size approximately $1$, and its profit is approximately $\frac{k+\Gamma}2$ (up to negligible terms).  Routine inspection of the cases of odd and even $k-\Gamma$ yields that the competitive ratio is at least $$\frac{3k+\Gamma-(k-\Gamma \mod 2)}{2k+2\Gamma} \ . $$

If $\Gamma \geq k\cdot \frac{\sqrt{33}-5}2$, we have $\frac{2k}{2k-\Gamma} \geq 0.75+\sqrt{33}/12\approx 1.228713$, and otherwise, $\frac{3k+\Gamma}{2k+2\Gamma} \geq 0.75+\sqrt{33}/12\approx 1.228713$. Since for large $k$, the fraction $\frac{1}{2k+2\Gamma} \leq \frac {1}{2k}$ tends to zero, the lower bound follows.

In the cases $k=2,3,4,5,6,7,8,9,10$ we get: $$\frac 43, \ \frac 54, \ \frac 65, \ \frac 54, \ \frac 54, \ \frac {11}{9}, \ \frac{16}{13}, \ \frac 54, \mbox{ \ and \ } \frac{16}{13} \ , $$ by testing all values of $0 \leq \Gamma \leq \lfloor \frac k2 \rfloor$, the bounds of the first case and the suitable bound for the second case.

Again, we note that as all items have sizes almost equal to $\frac{1}{2}$, the proofs works also for unit profits (where these profits are approximately twice as large as the proportional profits).
\end{proof}

%\date{}

%{\bf{\large A lower bound for unit removable knapsack}}
%
%\begin{theorem}
%``Roughly speaking'', any deterministic algorithm for $k \geq 2$ bins has competitive ratio of at least $1.228713$.
%\end{theorem}
%\begin{proof}
%The construction is as before, the only change is in the profits of the algorithms.
%
%In the first input, an optimal solution has profit $2k$. The algorithm has at least $\Gamma$ bins with just one item, so its profit is at most $2k-\Gamma$. The competitive ratio is again $\frac{2k}{2k-\Gamma}$ (it is unchanged because all input items have similar sizes, even though some are smaller than $\frac 12$ and some of them are larger).
%
%In the second input, an optimal solution has profit $k+\Gamma+\lfloor \frac{k-\Gamma}2 \rfloor$, and the algorithm cannot have more than $\Gamma$ bins with two items, so its profit is $k+\Gamma$. For even values of $k-\Gamma$, we get a competitive ratio of at least $\frac{3k+\Gamma}{2k+2\cdot\Gamma}$, and for odd values we get  $\frac{3k+\Gamma-1}{2k+2\cdot\Gamma}$.
%
%In the cases $k=2,3,4,5,6,7$ we get: $\frac 43$, $\frac 54$, $1.2$, $\frac 54$,  $1.25$, and $\frac {11}{9}$.
%
%For large even values of $k$, the lower bound on the competitive ratio tends to $0.75+\sqrt{33}/12\approx 1.228713$ which is achieved for $\Gamma = \frac{\sqrt{33}-5}2$.
%\end{proof}

\subsection{Randomized online algorithms for the proportional case}
In this section we provide an alternative easier construction, which in general provides weaker bounds, though with the exception of two values of $k$, for which it improves the result of Theorem~\ref{122} (which was for deterministic online algorithms).  Unlike that theorem, this one applies to randomized algorithms as well.  Another difference is that here the input sequence is not ``accommodating'', i.e.,
the optimal offline solution does not always pack all the items. The advantages of this construction are that it is fairly simple, and that it provides a lower bound for the case of randomized online algorithms.

\begin{theorem}
Any randomized online algorithm for $k \geq 2$ bins has competitive ratio of at least $1.2$.
\end{theorem}
\begin{proof}
Let $\eps>0$ be a very small constant. The input starts with $2k$ items of each of the two sizes: $\frac 23-\eps$, $\frac 13+3\eps$. Since items are removable, and since there are sufficiently many items, we can assume that every bin will either have one item of size $\frac 23-\eps$ or two items of sizes $\frac 13+3\eps$.

Let $X$ be the expected number of bins with one item of size $\frac 23-\eps$.
The input continues in one of two ways. The first one is $k$ items of sizes $\frac 13+\eps$, and the second one is $k$ items of sizes $\frac 23-3\eps$.  The optimal offline solution has $k$ full bins in either case.  Consider the online algorithm.
In the first case, its best approach is to add one item to each bin with an item of size $\frac 23-\eps$. These bins are full, so there is no better packing for them.  For each bin of the other kind, even if some replacements are made, the total size of the items it holds is at most $\frac 23 + 6\eps$.  By linearity of expectation, the expected profit of the algorithm is $k-\frac{k-X}3$ (letting $\eps$ tend to zero and neglecting those terms).

In the other case, there is no reason for the algorithm to replace items of size approximately $\frac 23$.  (Besides, such replacements change only the negligible $\eps$ terms.)  Every bin without such an item can become full when the algorithm replaces one item of size $\frac 13+3\eps$ with the item of size $\frac 23-3\eps$). The expected profit of the algorithm is $k-\frac X3$.

The threshold for $X$ is $\frac k2$. The first input is used if $X \leq \frac k2$, and the second one is used otherwise, if $X > \frac k2$. The competitive ratio is at least $1.2$.

For odd values of $k$,  in the deterministic case we can use the integrality of $X$. The first input is used if $X \leq \frac {k-1}2$, and the second one is used otherwise, i.e., for $X \geq \frac {k+1}2$. The ratio is at least $$\frac{k}{k-\frac{k+1}6}=\frac{6k}{5k-1}>1.2 \ . $$
For $k=3,5,7,9$ we get lower bounds of $\frac 97$, $1.25$, $\frac{21}{17}$, and $\frac{27}{22}$, which improves slightly upon the lower bound from Theorem~\ref{122} for $k=3$ and $k=7$.
\end{proof}

\section{The online minimum peak appointment scheduling problem}
In this section we study the problem MPAS. On high level, the aim in this problem is the same as in the multiple knapsack problem, i.e., to pack the items efficiently,
which possibly explains why our results are similar in spirit and techniques.  However, the setup is rather different, and it is a minimization problem rather than maximization.
Here, every item has to be packed, and the goal is to minimize the number of the bins.  The algorithm is not required to specify the bin for packing, but it has to specify the item's position in any bin it will eventually be packed in, i.e.,
for an arriving item of size $\gamma$ (where $0<\gamma\leq 1$), the algorithm has to specify an interval of the form $[x,x+\gamma)$ such that $0 \leq x\leq 1-\gamma$.

\subsection{Warm-up: deterministic online algorithms for MPAS}
We start with a simple construction for deterministic algorithms that uses the adaptive item sizing technique.
Afterwards, we improve it into a lower bound for randomized algorithms.

\begin{theorem}
The asymptotic competitive ratio for deterministic online algorithms for MPAS is at least $1.25$.
\end{theorem}
\begin{proof}
The input consists of one or two phases, depending on the algorithm.  In the first phase, for sufficiently large $N>0$, $12N$ items arrive with sizes chosen adaptively using $\alpha=\frac 13-2\eps$ and $\beta=\frac 13-\eps$ for an arbitrarily small value $\eps>0$.  We define the partition into smallish and largish items now. Those items whose intervals after packing contain the point $\frac 12$ are classified as smallish, and the remaining items (whose intervals do not contain the point $\frac 12$) are classified as largish.  Recall that the adaptive sizing guarantees that there exists a threshold $\theta \in (\alpha, \beta)$ such that smallish items have sizes in $[\alpha,\theta)$ and largish items have sizes in $(\theta,\beta]$.  Let $Q$ denote the number of smallish items.
%In the proof, we sometimes assign an interval that is slightly too long in our optimal solutions.

We further classify the largish items as either ``low'' or ''high'', depending on whether their intervals lie completely
to the left or to the right of the point $\frac{1}{2}$, if the positions are defined on the horizontal axis.  Note that all high items contain the point $\frac 34$ in their intervals and similarly all low items contain the point $\frac 14$.  If $Q \geq 5N$ or $Q \leq 2N$, then a solution that distributes the items evenly, i.e., places $4N$ items in each of the intervals $$\left[0,\frac 13\right), \mbox{ \ \ \ \ } \left[\frac 13,\frac 23\right), \mbox{ \ \ \ \ } \left[\frac 23,1\right) \ , $$ and thus has cost of $4N$, proves that the algorithm's competitive ratio is at least $\frac54$: For $Q \geq 5N$, there are $Q \geq 5N$ smallish items, all containing the point $\frac{1}{2}$, whereas for $Q \leq 2N$, there are must be at least $5N$ low or at least $5N$ high items (by the pigeonhole principle), which then all contain the point $\frac{1}{4}$ or $\frac{3}{4}$ respectively.
Otherwise, when $Q \in (2N,5N)$, there is a second phase, which depends on $Q$'s relation to $3N$.

If $Q\geq 3N$, $12N$ items of size $\frac 23$ each are issued. An optimal offline solution places the first phase items in the interval $[0,\frac 13)$, and it places the second phase $12N$ items in the interval $[\frac 13,1)$, yielding optimal cost of $12N$. For the algorithm, all $12N$ second phase items must have the point $\frac 12$ as an internal point of their intervals, as do the $Q$ smallish items, so the cost of the algorithm is at least $15N$.  Thus the asymptotic competitive ratio is at least $1.25$ in this case.

Finally, if $Q\leq 3N$, $Q'$ items of size $1-\theta$ are issued, where $Q'$ is divisible by $3$ and $Q-2 \leq Q' \leq Q$. Note that all low items have the point $\theta$ as an internal point of their intervals, and similarly, all high items have the point $1-\theta$ as an internal point. %%%Moreover, since $1-\theta >\frac 23$, all items of size $1-\theta$ have the point $\frac 12$ as an internal point.
We find that the second phase items have a common point with every interval of the low and high items from the first phase. By the pigeonhole principle, there are at least $\frac{12N-Q}2$ low items or at least this many high items. Thus, for at least one of the points $\theta$ and $1-\theta$, there are at least  $\frac{12N-Q}2+Q'$ items whose intervals contain it, for a cost of at least $\frac{12N+Q'}{2}-1$ for the algorithm. An optimal offline solution has $Q'$ intervals of $[0,\theta)$ for smallish items, $Q'$ intervals of $[\theta,1)$ for items of size $1-\theta$, and for the remaining $12N-Q'$ items (where this number is divisible by $3$) there are $4N-\frac{Q'}3$ intervals of every form out of $[0,\frac 13)$, $[\frac 13,\frac 23)$, $[\frac 23,1)$. Note that all smallish items have sizes not exceeding $\theta$, and the intervals assigned to such items by an optimal offline solution are sometimes slightly too long (because they have lengths of $\theta$).

The cost of an optimal offline solution is therefore at most $\frac{12N+2Q'}3$. For large value of $N$, we can neglect the additive term and find a lower bound on the ratio $$\frac{(12N+Q')/2}{(12N+2Q')/3}$$ for $Q'\leq Q< 3N$. This ratio is indeed at least $\frac 54$ for $Q' < 3N$ since $$\frac{3(12+\frac{Q'}N)}{2(12+2\frac{Q'}{N})}$$ is a monotonically decreasing function of $\frac{Q'}{N}$, and for $Q'\leq 3N$ it is minimized for $\frac{Q'}{N}=3$, in which case $$\frac{(12N+Q')/2}{(12N+2Q')/3}=1.25 \ .$$

As in each of the four cases analyzed in the proof, the competitive ratio is at least $1.25$, possibly in the limit as $N$ grows to infinity, $1.25$ is a lower bound on asymptotic competitive ratio.
\end{proof}

\subsection{A lower bound for randomized algorithms for MPAS}

Now, instead of using an adaptive classification of items, we show how to use very small items instead.
This construction, which yields a superior bound, also applies to randomized algorithms.

\begin{theorem}
The deterministic and randomized asymptotic competitive ratios for the peak problem are at least $1.2691534$.
\end{theorem}
\begin{proof}
Let $N,M>0$ be large integers, such that $N$ is even, and $M$ is divisible by $N!$.

%The construction can be used for randomized algorithms almost as it is defined, but we will consider deterministic algorithms first.
To prove the lower bound for randomized algorithms, we use Yao's approach, where one considers the best deterministic online algorithm for a known
probability distribution over inputs.

Each input starts with a fixed prefix of $N \cdot M$ items of size $\frac 1N$ each. For every point $z$ such that $0\leq z<1$, define $f(z)$ to be the number of items whose assigned intervals contain the point $z$ as an interior point or the left endpoint of the interval. Since the length of every interval is $\frac 1N$, its contribution to the definite integral $\int_{0}^1 f(z)\ \textrm{d}z$ is $\frac 1N$, and therefore
\begin{equation*}\label{eq:f-integral}
\int_{0}^1 f(z)\ \textrm{d}z = \frac{MN}{N} = M \ .
\end{equation*}
The integration is possible since the number of discontinuity points of $f$ is at most $2MN$, i.e., a constant for every fixed pair $(N,M)$. In fact, $f$ is constant between every two consecutive discontinuity points, including the boundary points $0$ and $1$ among those.
This implies that we can find the total length of intervals where $f$ has the integer value $i$ for $0\leq i \leq MN$, which we denote by $\beta_i$.  Clearly, $\sum_{i=0}^{MN} \beta_i=1$. Imagine sorting the intervals with fixed values of $f$, so that, going from right to left along the $[0,1)$ interval, we have, in this order, a sequence of intervals, the $i+1$-th of which, where $0 \leq i \leq MN$, has length $\beta_i$ and associated value $i$.  This is captured by a non-increasing step function $g$ defined as follows: For every point $z$ where $0\leq z < 1$, let $g(z)$ the unique $i$ such that $$\sum_{j=i+1}^{MN} \beta_j \leq z \ \ \mbox{ and \ \ \ } \sum_{j=i}^{MN} \beta_j > z \ ; $$ the function $g$ is well-defined, since naturally $\sum_{j=MN+1}^{MN} \beta_j=0$.  Moreover, it holds that
\begin{equation}\label{eq:g-integral}
	\int_{0}^1 g(z)\ \textrm{d}z = \sum_{i=0}^{MN} i\beta_i = \int_{0}^1 f(z)\ \textrm{d}z = M \enspace.
\end{equation}
We further note that it follows from the definition of $g$ that any (left-closed) interval of length at least $\ell$ contains a point $z$ such that $g(z) \geq f(1-\ell)$, i.e., a point $z$ which is contained in at least $f(1-\ell)$ intervals assigned by the algorithm to the items from the input prefix.

%Given the function $g$ for a specific algorithm,
Next, the input may continue in one of many ways.  Specifically, consider an eventually fixed $t$ such that $1 \leq t \leq \frac N2-1$.  Then,
for every $q=t,t+1,\ldots,\frac N2-1$, there is an input $I_q$ which continues after the prefix with $\frac{MN}{q}$ items of length $1-\frac qN$. Since $q\leq  \frac N2-1$, we have  $$1-\frac qN \geq 1-\frac{ \frac N2-1}N=\frac 12+\frac 1N > \frac 12 \ . $$ An optimal offline solution assigns all these items the interval $[\frac qN,1]$, and it partitions the $MN$ items of size $\frac 1N$ from the prefix into $q$ subsets of $\frac{MN}{q}$, to be assigned the intervals $[\frac{j-1}N,\frac jN)$ for $j=1,2,\ldots, q$, i.e., all items from a $j$-th subset are assigned the $j$-th interval.  Clearly, the cost of such solution is $\frac{MN}{q}$, i.e.,
\begin{equation}\label{eq:opt-q}
	\text{OPT}(I_q) = \frac{MN}{q} \enspace.
\end{equation}

%We discuss the lower bound for deterministic algorithms first.
As for the algorithm, no matter what intervals it assigned to the items of size $1-\frac qN$, they must all contain the interval $J_q=[\frac qN, 1-\frac qN)$, whose length is $1-\frac{2q}{N}$.  Thus, by aforementioned properties of the function $g$, there is a point $z \in J_q$ such that
$f(z) \geq g\left(\frac{2q}{N}\right)$, which implies that
\begin{equation}\label{eq:alg-q}
	\text{ALG}(I_q) \geq g\left(\frac{2q}{N}\right) + \frac{MN}{q} \enspace.
\end{equation}
In addition, we consider the prefix of items by itself, i.e., with no further items released, and denote such instance $I_{N/2}$.  The optimal offline solution partitions items into $N$ subsets and uses all intervals $[\frac {j-1}N,\frac jN)$ for $q\leq j \leq N$, so
\begin{equation}\label{eq:opt-0}
	\text{OPT}(I_{N/2}) = M \enspace,
\end{equation}
while by definition and properties of the function $g$, for the online algorithm we have
\begin{equation}\label{eq:alg-0}
	\text{ALG}(I_{N/2}) \geq g\left(0\right)\enspace.
\end{equation}

Suppose that the algorithm is asymptotically $R$-competitive.  Then, for some additive constant $C$, the following inequality holds for any probability distribution $\{p_q\}_{q=t}^{N/2}$ over the instances $\{I_q\}_{q=t}^{N/2}$:
\begin{equation}\label{eq:R-comp}
	\mathbb{E}_q\left[\text{ALG}(I_q)\right] \leq R \cdot \mathbb{E}_q\left[\text{OPT}(I_q)\right] + C\enspace.
\end{equation}
Plugging in the upper bounds on OPT, i.e., \eqref{eq:opt-q} and \eqref{eq:opt-0}, as well as the lower bounds on ALG, i.e., \eqref{eq:alg-q} and \eqref{eq:alg-0}, we get
\begin{equation*}
	p_{\frac{N}{2}} \cdot g(0) +\sum_{q=t}^{\frac N2-1} p_q \left( g\left(\frac{2q}N\right) + \frac{MN}{q} \right) \leq R \left(p_{\frac{N}{2}} \cdot M +\sum_{q=t}^{\frac N2-1} p_q \frac{MN}q \right) + C \enspace,
\end{equation*}
which after moving the terms without $g$ to the right hand side becomes
\begin{equation*}
	p_{\frac{N}{2}} \cdot g(0) +\sum_{q=t}^{\frac N2-1} p_q \cdot g\left(\frac{2q}N\right) \leq R \cdot p_{\frac{N}{2}} \cdot M + \left(R-1\right) \sum_{q=t}^{\frac N2-1} p_q \frac{MN}q +C \enspace.
\end{equation*}
Letting $p_{\frac{N}{2}} = \frac{2t}{N}$ and $p_i=\frac{2}{N}$ for $t \leq i < N/2$, the left hand side becomes an upper bound on the integral of $g$ over $[0,1)$, so by~\eqref{eq:g-integral},
\begin{equation*}
	M = \int_{0}^{1} g(t)\ \textrm{d}t \leq \frac{2t}{N} \cdot g(0) +\frac{2}{N}\sum_{q=t}^{\frac N2-1} g\left(\frac{2q}N\right) \leq \frac 2N \cdot (t\cdot R\cdot M+\sum_{q=t}^{\frac N2-1} (R-1) \frac{N}q) +C \enspace.
\end{equation*}
Dividing by $2M$, the term $\frac{C}{2M}$ tends to $0$ for $M\to\infty$, so letting $\tau = \frac{t}{N}$, this inequality becomes
\begin{equation*}
	\tau\cdot R+(R-1)\sum_{q=\tau\cdot N}^{\frac N2-1}  \frac{1}q \geq \frac 12  \enspace.
\end{equation*}
With $N$ growing to infinity, the sum $\sum_{q=\tau\cdot N}^{\frac N2-1}  \frac{1}q$ tends to $ - \ln (2\tau)$.  Rearranging, we have
\begin{equation*}
	\left( R-1 \right) \left( \tau - \ln\left(2\tau\right) \right) \geq \frac{1}{2} - \tau \enspace,
\end{equation*}
where finally letting $\tau\approx 0.212072$ yields the desired lower bound on $R$.

\medskip

We note that the even though we did not specify the probability distribution over instances upfront, it is fixed, and in particular it does not depend on the deterministic online algorithm, which thus may know the distribution a priori, as stipulated by Yao's principle.
\end{proof}

\bibliographystyle{abbrv}

\bibliography{bprelaxed}

\end{document}